\documentclass[letterpaper, 10 pt, conference]{ieeeconf}  

\IEEEoverridecommandlockouts                              
\let\labelindent\relax

\usepackage{grffile}
\usepackage{amsmath,amssymb,amsbsy}
\usepackage{hyperref}
\usepackage{amsthm} 
\usepackage{color}
\usepackage{comment}
\usepackage{lipsum}
\usepackage{graphicx}%
\usepackage{subfigure}
\usepackage[normalem]{ulem}
\usepackage{enumerate}
\usepackage{epstopdf}
\usepackage{thmtools, thm-restate}
\newtheorem{theorem}{Theorem}[section]
\newtheorem{lemma}[theorem]{Lemma}

\newtheorem{corollary}[theorem]{Corollary}	
\newtheorem{remark}[theorem]{Remark}	
\newtheorem{problem}[theorem]{Problem}

\title{\LARGE \bf
Controllability to Equilibria of the 1-D Fokker-Planck Equation with Zero-Flux Boundary Condition
}

\author{Karthik Elamvazhuthi, Hendrik Kuiper, and Spring Berman
\thanks{This work was supported by National Science Foundation (NSF) Award CMMI-1436960 and by 
ONR Young Investigator Award N00014-16-1-2605.}
\thanks{Karthik Elamvazhuthi and Spring Berman are with the School for Engineering of
Matter, Transport and Energy, Arizona State University, Tempe, AZ, 85281
USA {\tt\small \{karthikevaz, Spring.Berman\}@asu.edu}. Hendrik Kuiper is with the School of Mathematical 
and Statistical Sciences, Arizona State University, Tempe, AZ, 85281 USA {\tt\small \{kuiper@asu.edu\}}%
}}

\begin{document}

\maketitle
\thispagestyle{empty}
\pagestyle{empty}

\begin{abstract}
We consider the problem of controlling the spatiotemporal probability distribution of a robotic swarm that evolves according to a reflected diffusion process, 
using the space- and time-dependent drift vector field parameter as the control variable. 
In contrast to previous work on control of the Fokker-Planck equation, a zero-flux boundary condition is imposed on the partial differential equation that governs the swarm probability distribution, and only bounded vector fields are considered to be admissible as control parameters. Under these constraints, we show that any initial probability distribution can be transported to a target probability distribution under certain assumptions on the regularity of the target distribution. In particular, we show that if the target distribution is (essentially) bounded, has bounded first-order and second-order partial derivatives, and is bounded from below by a strictly positive constant, then this distribution can be reached exactly using a drift vector field that is bounded in space and time. Our proof is constructive and based on classical linear semigroup theoretic concepts. 
\end{abstract}

\section{INTRODUCTION}
\label{sec:intro}
In recent years, there has been much work on the modeling and control of swarms of homogeneous agents using mean-field models. These mean-field models are typically defined by a system of partial differential equations (PDEs) that describe how an initial probability measure is pushed forward under the action of an ordinary differential equation (ODE) or a stochastic differential equation (SDE).  In this context, spatial information on the agent positions is modeled using probability measures, and the mean-field control problem is to design parameters of the system of PDEs so that these probability measures evolve in a desirable manner. 

This perspective has led to a number of works on optimal control of PDEs with the goal of optimizing swarm behavior. In \cite{milutinovic2006modeling}, the authors use a maximum principle for control of infinite-dimensional systems to design optimal switching parameters that achieve target swarm densities. The work in \cite{annunziato2010optimal} addresses the problem of optimal control of the Fokker-Planck equation to ensure that its solution tracks a predefined time-dependent reference density. 
A similar approach is used for mean-field games and mean-field type controls in \cite{bensoussan2013mean,gomes2014mean}, which consider controlled versions of Vlasov-Mckean type SDEs whose coefficients are coupled to the distribution of the stochastic process. While most prior work has considered the case where there is noise in the agent dynamics, there has also been some recent work on the noise-less case where each agent evolves according to an ODE \cite{bongini2015mean,fornasier2014mean} rather than an SDE.

The literature on control of mean-field models has mostly focused on the synthesis of optimal control strategies. However, there has been some work on questions of stabilization and controllability of such models. For example, the design of output feedback laws for designing globally stable invariant distributions of stochastic processes was considered in \cite{mesquita2012jump}. Control of swarm protocols governed by the kinetic Cucker-Smale model was addressed in \cite{piccoli2015control} to produce flocking behavior. The Benamou-Brenier fluid dynamic formulation of optimal transport problems \cite{benamou2000computational} can also be interpreted as a problem of optimal control and controllability of the continuity equation. See also the work by Brockett \cite{brockett2012notes} on related problems on the control of Liouville equations. Closer to the work in this paper are the studies \cite{blaquiere1992controllability,dai1991stochastic,porretta2014planning} on controllability of Fokker-Planck equations. In \cite{blaquiere1992controllability} and \cite{dai1991stochastic}, it was proven that solutions of the Fokker-Planck equation evolving on $\mathbb{R}^n$ can be controlled to a large class of target distributions. This result has been extended to the case where the agent dynamics are governed by general linear SDEs and the initial and target distributions are Gaussian \cite{chen2016optimalI,chen2016optimalII}. Poretta \cite{porretta2014planning} considered the problem of controllability of the Fokker-Planck equation when the solutions evolve on a torus, along with the well-posedness of an associated mean-field game problem. 

In this work, we consider controllability of the Fokker-Planck equation on a bounded one-dimensional domain with zero-flux boundary condition. This problem is of practical importance for swarm robotic applications in which the agents' spatiotemporal behavior can be modeled by the Fokker-Planck equation and the agents are constrained to evolve in a bounded domain \cite{elamvazhuthi2016coverage,foderaro2014distributed,milutinovic2006modeling,prorok2011multi}.  Specifically, we prove (see Theorem \ref{maintheo}) that when the target probability density is sufficiently regular, it can be reached from any square-integrable initial probability distribution in a given finite time using bounded control inputs. Our arguments are entirely based on linear operator semigroup theoretic concepts and make a straightforward exploitation of the fact  that the exponential convergence rate to the target equilibrium can be increased arbitrarily using an appropriate density-feedback law. This approach differs from those in 
  \cite{blaquiere1992controllability} and \cite{dai1991stochastic}, which are based on probabilistic and stochastic control theoretic concepts, and the approach in \cite{porretta2014planning}, which uses observability inequality type arguments. Although it might be possible to adapt these methods for our scenario, the constraint of the vector field being bounded is not imposed in these works, which is more relevant for practical scenarios of interest to us. On the other hand, unlike the works \cite{dai1991stochastic,blaquiere1992controllability,porretta2014planning}, we do not address any issue regarding the optimality of the control laws. We note that although we restrict our analysis in this paper to the case of 1-D domains for the sake of simplicity, our approach has natural extensions to the more practical multi-dimensional case, which will be the subject of our future work.

\section{PROBLEM FORMULATION}
Consider a swarm of $n$ agents deployed on the one-dimensional domain $[0,1]$. The position of each agent, indexed by $i \in \lbrace 1,2,...,n \rbrace$, evolves according to a stochastic process $Z_i(t) \in [0,1]$, {where $t$ denotes time. Since the random variables that correspond to the dynamics of each agent are independent and identically distributed, we can drop the subscript $i$ and define the problem in terms of a single stochastic process, $Z(t) \in [0,1]$. The deterministic motion of each agent is defined by a vector field $v(x,t) \in \mathbb{R}$, where $x \in [0,1]$. This motion is perturbed by the Wiener process $W(t)$, which models noise. This process can be a model for stochasticity arising from inherent sensor and actuator noise. Alternatively, noise could be actively programmed into the agents' motion 
to implement more exploratory agent behaviors and to take advantage of the smoothening effect of the process on the agents' probability densities. 
Given the parameter $v(x,t)$, the stochastic process $Z(t)$ satisfies the following SDE \cite{tanaka1979stochastic}:
\begin{eqnarray}
\label{eq:SDE}
dZ(t) &=& ~ v(Z,t)dt + dW(t) + d\psi(t),  \\ \nonumber
Z(0) &=& ~ Z_0,
\end{eqnarray}
where $d\psi(t) \in \mathbb{R}$ is called the {\it reflecting function}, a stochastic process that constrains $Z(t)$ to the domain $[0,1]$.   

\begin{problem} 
\label{prob:SDEctrl}
Given $T>0$ and $f:[0,1] \rightarrow \mathbb{R}^{+}$ such that $\int_0^1 f(x)dx = 1$, determine if there exists a feedback control law $v: [0,1] \times [0,T] \rightarrow \mathbb{R}$ such that the process \eqref{eq:SDE} satisfies $\mathbb{P}(Z(T) \in \Gamma) = \int_\Gamma f(x)dx$ for each Borel subset $\Gamma \subset [0,1]$.
\end{problem}

The {\it Kolmogorov forward equation} corresponding to the SDE  \eqref{eq:SDE} is given by \cite{pilipenko2014introduction}:
\begin{eqnarray}
\nonumber
&y_t = y_{xx} -  (v y)_x & in  ~~ (0,1) \times [0,T] \\ \nonumber
&y(\cdot,0) = y_0 & in ~~ (0,1)  \\ \nonumber
&(y_x-vy)(0,\cdot)= (y_x-vy)(1,\cdot) = 0  & in ~~ [0,T] \label{eq:Mainsys1} \\
\end{eqnarray}
The solution $y(x,t)$ of this equation represents the probability density of a single agent occupying position $x \in [0,1]$ at time $t$, or alternatively, the density of a population of agents at this position and time.
The PDE \eqref{eq:Mainsys1} is related to the SDE \eqref{eq:SDE} by the relation $\mathbb{P}(Z(t) \in \Gamma) = \int_{\Gamma} y(x,t)dx$ for all $t \in [0,T]$ and all measurable $\Gamma$.
Problem \ref{prob:SDEctrl} can be reframed in terms of equation \eqref{eq:Mainsys1} as a PDE controllability problem as follows:
\begin{problem}
\label{prob:ADEctrl}
Given $T>0$, $y_0 :[0,1] \rightarrow \mathbb{R}^{+}$, and $f:[0,1] \rightarrow \mathbb{R}^{+}$ such that $\int_0^1 y_0(x)dx =\int_0^1 f(x)dx = 1$, determine whether there exists a space- and time-dependent parameter $v: [0,1] \times [0,T] \rightarrow \mathbb{R}$ such that the solution $y$ of the PDE  \eqref{eq:Mainsys1} satisfies $y(\cdot,T) = f$.
\end{problem}

\section{PRELIMINARIES AND NOTATION}
We define $L^2(0,1)$ as the Hilbert space of square-integrable real-valued functions over the unit interval, $(0,1) \subset \mathbb{R}$. The Hilbertian structure of $L^2(0,1)$ is induced by the standard inner product, $\langle \cdot , \cdot \rangle_{2} :  L^2(0,1) \times L^2(0,1) \rightarrow \mathbb{R}$, given by:
\begin{equation}
\langle p,q\rangle_{2} = \int_{0}^1 p(x)q(x)dx
\end{equation}
for each $p,q\in L^2(0,1)$. The norm $\|\cdot\|_{2}$ on the space $L^2(0,1)$ is defined as
\begin{equation}
\|p\|_{2} = \langle p, p\rangle^{1/2}_{2}
\end{equation}
for all $p \in L^2(0,1)$. For a function $r \in L^2(0,1)$ and a given constant $c$, we write $r \geq c$ to imply that $r(x)\geq c$ for almost every $x \in (0,1)$.

We define the Sobolev space $H^1(0,1) = \big \lbrace z \in L^2(0,1): z_{x} \in L^2(0,1) \big \rbrace$. We equip this space with the usual Sobolev norm $\|\cdot\|_{H^1}$, given by
\begin{equation}
\|p\|_{H^1} = \Big( \|p\|^2_{2} + \|p_x\|^2_{2}\Big)^{1/2} \nonumber
\end{equation}
for each $p \in H^1(0,1)$. We will also need the space $H^2(0,1) = \big \lbrace z \in H^1(0,1): z_{xx} \in L^2(0,1) \big \rbrace$, which will be equipped with the norm $\|\cdot\|_{H^2}$ given by
\begin{equation}
\|p\|_{H^2} = \Big( \|p\|^2_{H^1} + \|p_{xx}\|^2_{2}\Big)^{1/2} \nonumber
\end{equation}
for each $p \in H^1(0,1)$. We define $L^{\infty}(0,1)$ as the space of essentially bounded measurable functions on $(0,1)$. The space $L^{\infty}(0,1)$ is equipped with the norm
\begin{equation}
\|z\|_{\infty} = ess ~ sup_{x \in (0,1)} |z(x)|, 
\end{equation}
where $ess ~ sup_{x \in (0,1)}(\cdot)$ denotes the {\it essential supremum} attained by its argument over the interval $(0,1)$.
For a given $a \in L^{\infty}(0,1)$, $L^2_a(0,1)$ will refer to the set of all functions $p$ such that
\begin{equation}
\int_0^1|p(x)|^2a(x)dx< \infty.
\end{equation}
The space $L^2_a(0,1)$ is a Hilbert space with respect to the weighted inner product $\langle \cdot , \cdot \rangle_{a} :  L^2_a(0,1) \times L^2_a(0,1) \rightarrow \mathbb{R}$, given by:
 \begin{equation}
\langle p, q\rangle_{a} = \int_{0}^1 p(x)q(x)a(x)dx
\end{equation}
for each $f,g \in L^2(0,1)$.
We also define the Sobolev spaces
\begin{equation}
W^{1,\infty}(0,1) = \big \lbrace z \in L^{\infty}(0,1): z_{x} \in L^{\infty}(0,1) \big \rbrace,
\end{equation}
\begin{equation}
W^{2,\infty}(0,1) =  \big \lbrace z \in W^{1,\infty}(0,1): z_{xx} \in L^{\infty}(0,1) \big \rbrace.
\end{equation}
The space $C([0,T];L^2(0,1))$ consists of all continuous functions $u:[0,T] \rightarrow L^2(0,1)$ for which \[ \| u\|_{C([0,T];L^2(0,1))} ~:=~ \max_{0 \leq t \leq T} \|u(t)\|_2 ~<~ \infty . \]  

We will need an appropriate notion of a solution of the PDE \eqref{eq:Mainsys1}. Toward this end, let $A$ be a closed linear operator on a Hilbert space $X$ with domain $\mathcal{D}(A)$. For a given time $T>0$, a {\it mild solution} of the ODE
\begin{equation} 
\dot{u}(t) = Au(t); \hspace{2mm} u(0) = u_0 \in L^2(0,1)
\end{equation}
is a function $u \in C([0,T];L^2(0,1))$ such that $u(t) = u_0 + A\int_0^tu(s)ds$ for each $t \in [0, T]$. Under appropriate conditions satisfied by $A$, the mild solution is given by a semigroup of linear operators, $(T(t))_{t\geq 0}$, that are {\it generated} by the operator $A$. That is, the solution of the above ODE is given by $u(t)= T(t)u_0$ for each $t \in [0,T]$.

The differential equations that we analyze in this paper will be non-autonomous in general. Hence, we must adapt the notion of a mild solution to these types of equations. Let $A_i$ be a closed linear operator with domain $\mathcal{D}(A_i)$ for each $i \in \mathbb{Z}_+$. For a certain time interval $[0,T]$, a piecewise constant family of operators is given by a map, $t \mapsto A(t)$, for which there exists a partition of $[0,T] = \cup_{i \in \mathbb{Z}_+} [a_i,a_{i+1})$ such that $a_i \leq a_{i+1}$ for each $i \in \mathbb{Z}_+$ and $A(t) = A_i$ for each $t \in [a_i,a_{i+1})$. Then a mild solution of the ODE
\begin{equation} 
\label{eq:nonaut}
\dot{u}(t) = A(t)u(t); \hspace{2mm} u(0) = u_0 \in X
\end{equation}
is a function $u \in C([0,T];X)$ such that 
\begin{equation} 
\label{eq:mldsol}
u(t) = u_0 + \sum_{i \in  \mathbb{Z}_+}A_i\int_{min \lbrace t,a_{i} \rbrace}^{min \lbrace t,a_{i+1} \rbrace} u(s)ds
\end{equation}
for each $t \in [0, T]$. There is in fact a more general notion of mild solutions that arise from two-parameter semigroups of operators generated by time-varying linear operators. However, the definition in \eqref{eq:mldsol} will be sufficient for our purposes, since one can construct solutions of the ODE \eqref{eq:nonaut} by treating it as an autonomous system in each time interval $[a_i,a_{i+1})$ and patching these solutions together to obtain the solution $u$. Note that the mild solution is defined with respect to an operator $A$ or  collection of operators $A(t)$; when we refer to such a solution, the associated operator(s) will be clear from the context.


\section{CONTROLLABILITY ANALYSIS}

In this section, we prove our main result, Theorem \ref{maintheo}, as a solution to Problem \ref{prob:ADEctrl}. 
However, we first note the following result on exponential stabilizability of equilibrium distributions which is also useful from the point of view of multi-agent control problems and gives an 'approximate' result to the controllability problem \ref{prob:ADEctrl}. Particularly, this theorem gives a candidate time-independent vector field if one desires convergence to a given target distribution at an exponential rate. This is similar to our previous work \cite{elamvazhuthi2016coverage}, where we used a spatially inhomogenous diffusion coefficient to stabilize desired probability densities. However, one needs to assume higher regularity on the target distributions in the following scenario than the work in \cite{elamvazhuthi2016coverage}.

\begin{theorem}
\label{asymp}
Let $y_0 \in L^2(0,1)$ and $f \in W^{1,\infty}(0,1)$ such that $y_0 \geq 0$,  $f \geq k> 0$ for some strictly positive constant $k$ and $\int_{0}^1 y_0(x)dx = \int_{0}^1f(x)dx = 1$. Suppose $v(x,t) =  f_x(x)/f(x)$  for all $t \in [0, \infty)$ and almost every $x \in (0, 1)$ in the PDE, \eqref{eq:Mainsys1}, then a unique mild solution $y \in C([0,\infty);L^2(0,1))$ of the PDE exists and satisfies the estimate
\begin{equation}
\|y(\cdot, t) - f\|_{2} \leq M e^{-\lambda t} \|y_0 - f\|_{2} 
\end{equation}
for some positive constants, $M \geq 1$, $\lambda$ and all $t \in [0, \infty)$.
\end{theorem}
The above result is a straightforward corollary of a result well-known in mathematics and physics literature; 
namely, that if $v$ is the gradient of a potential function $\phi$, then $e^{\phi}$ is the unique stationary distribution of the stochastic process \eqref{eq:SDE}. From this, if one observes that $f_x/f = (ln ~ f)_x$, then it follows that any weakly differentiable function with a strictly positive lower bound can be written as an exponential distribution, $f = e^{(ln ~ f)}$. Hence, any function, $f$, that is at least once weakly differentiable can be then designed to be the equilibrium distribution of the Fokker-Planck equation by choosing the vector field to be the gradient of the potential function, $ln f$. See remarks in \cite{markowich2000trend},  \cite{ambrosio2009existence} and \cite{bakry2013analysis}. There are multiple ways to establish the stability result. One such approach is using optimal transport methods. Using this approach the above result follows when certain convexity conditions on $f$ are satisfied and the PDE, \eqref{eq:Mainsys1}, is framed as a gradient flow for an appropriate energy functional on the 2-Wasserstein space \cite{villani2008optimal}. This is shown in \cite{markowich2000trend} and \cite{ambrosio2009existence}. For more general $f$ as stated in the above theorem, the result follows using classical functional analytic methods, which might not sufficient when the domain of the PDE is unbounded or infinite dimensional as in the works  \cite{markowich2000trend} and  \cite{ambrosio2009existence} respectively. Particularly, the operator $(\cdot)_{xx} - (\cdot ~ \frac{f_x}{f} )_x$ with the zero-flux boundary condition, is a self-adjoint, negative semi-definite, closed and densely defined operator on the weighted space $L^2_f(0,1)$. $L^2_f(0,1)$ is isomorphic to the space $L^2(0,1)$ due to the upper and lower bounds on $f$. Hence, the spectrum of the generator,  $(\cdot)_{xx} - (\frac{f_x}{f} )_x$, lies on the real line and is bounded above by the principal eigenvalue. $0$ is an eigenvalue that has a positive eigenvector, $f$.  This implies that $0$ is that principal eigenvalue.  Principal eigenvalues of elliptic operators are simple and therefore the the rest of the spectrum lies to the left of some negative number, $-\lambda$. From this the result on expoenential stability of the steady state solution $f$ follows. See also \cite{bakry2013analysis} for more details to establish the spectral gap, $\lambda$, using Poincare inequalities. 

Next, we collect some preliminary results that will be needed to prove our main theorem.

\begin{restatable}[]{theorem}{fta}
\label{thm:FTA}
Consider the PDE, 
\begin{eqnarray}
\nonumber
&y_t = ay_{xx} &~~ in  ~~ (0,1) \times [0,T] \\ \nonumber
&y(\cdot,0) = y_0 &~~ in ~~ (0,1)  \\ \nonumber
&y_x(0,\cdot) =y_x(1,\cdot)= 0  &~~ in ~~ [0,T] \\ 
\label{eq:varsys1}
\end{eqnarray}
Let $a \in W^{2,\infty}(0,1)$ be such that $a \geq e$ for some strictly positive constant $e$. Let $y_0 \in L^2(0,1)$. Then a unique mild solution $y \in C([0,T];L^2(0,1))$ of the above PDE exists. Additionally, if there exists a positive constant, $c>0$, such that $y_0
 \geq c$, then $y(\cdot,t) 
 \geq c$ for each $t \in [0, \infty)$.
\end{restatable}
\begin{proof}
See appendix.
\end{proof}

Now, we make some definitions which will be used subsequently. Given $a \in W^{2,\infty}(0,1)$ such that $a  \geq e$ for some strictly positive constant $e$ we define the operator, $A_a :\mathcal{D}(A_a)  \rightarrow  L^2(0,1)$, by 
\begin{equation}
\label{eq:Clpgen}
A_au = (au)_{xx}
\end{equation} 
for each $u \in {D}(A_a) =\lbrace w \in H^2(0,1); (aw_x)(0)=(aw)_x(1)=0 \rbrace$. 

\begin{corollary}
\label{CorII}
Let $a \in W^{2,\infty}(0,1)$ be such that $a \geq e$ for some strictly positive constant $e$. Consider the PDE, 
\begin{eqnarray}
\nonumber
&y_t = (ay)_{xx} &~~ in  ~~ (0,1) \times [0,T] \\ \nonumber
&y(\cdot,0) = y_0 &~~ in ~~ (0,1)  \\ \nonumber
&(ay)_x(0,\cdot) =(ay)_x(1,\cdot)= 0  &~~ in ~~ [0,T] \\ 
\label{eq:cllp1}
\end{eqnarray}

Let $y_0 \in L^2({\Omega})$. Then a unique mild solution $y \in C([0,T];L^2(0,1))$ of the above PDE exists. Additionally, if there exists a positive constant, $c>0$, such that $y_0 
 \geq c$, then $y(\cdot,t) 
 \geq d$ for each $t \in [0, \infty)$, for some strictly positive constant $d$.
\end{corollary}
\begin{proof}
Both the existence of mild solutions of the PDE, \eqref{eq:cllp1}, and the lower bound on solutions follow from theorem \ref{thm:FTA} by noting that the coordinate transformation, $u=ay$, transforms the PDE \eqref{eq:cllp1}, to the PDE,
\begin{eqnarray}
\nonumber
&u_t = au_{xx} &~~ in  ~~ (0,1) \times [0,T] \\ \nonumber
&u(\cdot,0) = u_0 &~~ in ~~ (0,1)  \\ \nonumber
&u_x(0,\cdot) =u_x(1,\cdot)= 0  &~~ in ~~ [0,T] \\ 
\end{eqnarray}
\end{proof}

In the following lemma, we will establish some estimates on the rate of convergence of the solution, $y$, of the PDE  \eqref{eq:cllp1}, assuming the initial condition is regular enough. 

\begin{lemma}
\label{cvglma}
Let $y_0 \in \mathcal{D}(A_a)$ be such that $y_0 \geq 0$. Additionally, assume $a = 1/f$  , where $f \in W^{2,\infty}(0,1)$ such that $f \geq c>0$ for some constant  $c$, and $\int_0^1f(x)dx = \int_0^1y_0(x)dx$. Then the mild solution, $y \in C([0,\infty);L^2(0,1))$, of the PDE, \eqref{eq:cllp1}, satisfies $y(\cdot,t) \in \mathcal{D}(A_a) $ for each  $t \in [0,\infty)$. Moreover, the following estimates hold:
\begin{eqnarray}
\|y(\cdot,t)-f\|_{2} &\leq M_0 e^{-\lambda t} \\ \label{eq:expcn1}
\|A_ay(\cdot,t)\|_{2} &\leq M_1 e^{-\lambda t} \label{eq:expcn2}
\end{eqnarray}
for some strictly positive constants $M_0$, $M_1$ and $\lambda$.
\end{lemma}
\begin{proof}
The first estimate is just a restatement of \cite{elamvazhuthi2016coverage}[Theorem IV.4], where it was shown that $f$ is the eigenvector of $A_a$ corresponding to simple principal eigenvalue $0$, and the rest of spectrum is in the left-half complex plane, left to some negative number $-\lambda$. For this estimate, the assumption of the regularity of the initial condition is not required.

For the second estimate, we consider the space $(\mathcal{D}(A)_a, \|\cdot \|_g)$ where $\| \cdot \|_g$ is the graph norm given by 
\begin{equation}
\| z \|_g = \|z\|_2 +\|A_az\|_2
\end{equation}
for each $z \in \mathcal{D}(A_a)$. $A_a$ generates a semigroup of operators, $(T_1(t))_{t\geq 0}$, on the space, $(\mathcal{D}(A_a), \|\cdot \|_g)$. Let $(T(t))_{t\geq 0}$ be the semigroup of operators generated on the space, $L^2(0,1)$ by $A_a$. $T_1(t)$ is the restriction of the operator, $T(t)$, for each $t \in [0, \infty)$, on the space $(\mathcal{D}(A_a), \|\cdot \|_g)$. $T_1(t) = (I-A_a)^{-1}T(t)(I-A_a)$ for each $t \in[0,\infty)$. It follows that $T_1(t)$ and $T(t)$ are similar for each $t \in[0,\infty)$. Hence, the spectrum and growth bounds of $(T_1(t))_{t \geq 0}$ and $(T(t))_{t \geq 0}$ are the same. Hence, it follows that we have 
\begin{equation}
\|y(\cdot,t)-f\|_g \leq M_1 e^{-\lambda t} 
\end{equation}
for some $M_1>0$ and for all $t \in [0, \infty)$. This implies the estimate, \eqref{eq:expcn2}, since $A_a f=0$.
\end{proof}

Controllability of the PDE, \eqref{eq:Mainsys1}, will initially be established with some assumptions on  regularity conditions and lower bounds satisfied by the initial conditions (lemma \ref{prlma}). The next two results will help us relax these assumptions further ahead in the main theorem \ref{maintheo}.

\begin{restatable}[]{theorem}{ftak}
\label{thm:FTA2}
Let $y_0 \in L^2(0,1)$ and $f \in W^{2,\infty}(0,1)$ be such that  $f \geq k> 0$ for some positive constant $k$. Suppose $v(x,t) =  f_x(x)/f(x)$  for all $t \in [0, \infty)$ and almost every $x \in (0, 1)$ in the PDE, \eqref{eq:Mainsys1},
If there exists a positive constant, $c>0$, such that $y_0\geq c$, then the unique mild solution of the PDE satisfies the estimate $y(\cdot,t) 
 \geq d$ for each $t \in [0, \infty)$, and for some positive constant, $d$.

Moreover, we have that $\mathcal{D}(B_f) =\mathcal{D}(A_a)$, whenever $a = 1/f$ and $A_a$ is the operator defined in equation \eqref{eq:Clpgen} and $B_f :\mathcal{D}(B_f) \rightarrow L^2(0,1)$ is the operator given by
\begin{equation}
B_fu = u_{xx} - (\frac{f_x}{f}u)_x
\end{equation}
for each $u \in {D}(B_f) =\lbrace w \in H^2(0,1); (w_x-\frac{f_x}{f}w)(0)=(w_x-\frac{f_x}{f}w)(1)=0 \rbrace$. 
\end{restatable}

\begin{lemma}
\label{Dave}
Consider the heat equation with Neumann boundary condition, that is, $v \equiv 0$ in \eqref{eq:Mainsys1}. Let $y_0 \in L^2(0,1)$ be such that $y_0 \geq 0$. Let $y \in C([0,T];L^2(0,1))$ be the unique mild solution. Then for each $t \in (0,\infty)$ there exists a positive constant, $c_t>0$, such that $y(\cdot,t) \geq c_t$.
\end{lemma}
\begin{proof}
The solution $y$ of the PDE \eqref{eq:Mainsys1}, can be represented using the {\it Neumann heat kernel}, $K$. That is, there exists a measurable map, $K:(0,\infty) \times [0,1]^2   \rightarrow [0,\infty)$ such that the mild solution, $y$, is related to $K$ by the following relation:
\begin{equation}
y(x,t) =  \int_0^1K(t,x,z)y_0(x)dz
\end{equation}
for each $t \in (0,\infty)$ and almost every $x \in (0,1)$. From
 \cite{li1986parabolic}[Corollary 2.1],
 we know that the Neumann heat kernel, $K$, satisfies the following lower bound:
\begin{equation}
K(t,x,z) \geq \frac{1}{(4\pi t)^{1/2}}exp(\frac{-(x-z)^2}{4t})
\end{equation}
for each $t>0$ and almost every $x,z \in (0,1)$. From this the lower bound on $y(\cdot,t)$ follows. \end{proof}
\begin{lemma}
\label{prlma}
Let $y_0 \in \mathcal{D}(A_a)$ be such that $y_0 \geq c $ for some strictly positive constant, $c$. Suppose $f \in W^{2, \infty}(0,1)$ such that $f\geq s$  for some strictly positive constant $s$, $a = 1/f$ and $\int_0^1f(x)dx = \int_0^1y_0(x)dx$. Let $T = \sum_{n=1}^{\infty} \frac{1}{n^2}$ be the final time. Define the vector field $v$ by
\begin{equation}
\label{eq:fbctrl}
v(\cdot,t) =   \frac{y_x}{y}- \alpha m \frac{(a y)_x}{y}
\end{equation}
 with $a = 1/f$ in \eqref{eq:Mainsys1} whenever $t \in [\sum_{n=1}^{m-1} \frac{1}{n^2},\sum_{n=1}^{m} \frac{1}{n^2})$ and $m \in \mathbb{Z}_+$. Here, we define $\sum_{n=1}^{m} \frac{1}{n^2} = 0$ if $m=0$.

Then there exists an $\alpha >0$  such that $v \in L^{\infty}(0,T;L^{\infty}(0,T))$ and the mild solution $y$ of the PDE, \eqref{eq:Mainsys1}, satisfies $y(\cdot,T) = f$.
\end{lemma}
\begin{proof}
Substituting $v(\cdot,t) = \frac{y_x}{y}- \alpha m \frac{(a y)_x}{y}$ whenever $t \in [\sum_{n=1}^{m-1} \frac{1}{n^2},\sum_{n=1}^{m} \frac{1}{n^2})$, in \eqref{eq:Mainsys1}, it can be seen that the solution of the PDE,  \eqref{eq:Mainsys1} exists over each time interval $[0,\sum_{n=1}^{m} \frac{1}{n^2})$ for each $m \in \mathbb{Z}_+$. This is true because this solution can be constructed from mild solutions of the {\it closed-loop} PDE
\begin{eqnarray}
\nonumber
&\tilde{y}_t = \alpha m (a\tilde{y})_{xx} &~~ in  ~~ (0,1) \times [0,\frac{1}{m^2}) \\ \nonumber
&\tilde{y}(\cdot,0) = \tilde{y}_0=y(\cdot,\sum_{n=1}^{m-1} \frac{1}{n^2}) &~~ in ~~ (0,1)  \\ \nonumber
&(a\tilde{y})_x(0,\cdot) =(a\tilde{y})_x(1,\cdot)= 0  &~~ in ~~ [0,\frac{1}{m^2}) \\ 
\label{eq:pssol}
\end{eqnarray}
and we get the relation $y(\cdot, \sum_{n=1}^{m-1} \frac{1}{n^2} + j) = \tilde{y}(0,j)$ for each $j \in [0,\frac{1}{m^2})$ and each $m \in \mathbb{Z}_+$.  Then it follows from lemma \ref{cvglma} that
\begin{align*}
\label{cves1}
|y(\cdot,\sum_{n=1}^{m}\frac{1}{n^2})-f\|_{L^2(\Omega)} \leq  M_0 e^{-\alpha  \lambda \sum_{n=1}^{m}\frac{n}{n^2}} \\ \nonumber
=  M_0 e^{-\alpha \lambda \sum_{n=1}^{m}\frac{1}{n})}
\end{align*}

for each $m \in \mathbb{Z}_+$, for some strictly positive constants $M_0$ and $ \lambda$ independent of $m$. Since the summation $\sum_{n=1}^{m}\frac{1}{n}$ is diverging we have $y(\cdot,T) = f$ if the solution is defined over the interval $[0,T]$. By continuity of $y$ on $[0,T)$, it follows that $y \in C([0,T);L^2(0,1))$ and can be extended to a unique mild solution $y \in C([0,T];L^2(0,1))$ defined over the time interval $[0,T]$. 
 
It is additionally required to prove that $v \in  L^\infty(0,T;L^\infty(0,1))$. As will be shown further ahead, these results will follow if $\alpha>0$ is chosen to be large enough. More specifically, it will be established that if $\alpha$ is large enough we can get uniform bounds on $1/y(\cdot, t)$, $\alpha m(ay)_{x}(\cdot, t)$ and $y_{x}(\cdot, t)$ as $t$ is varied over the interval, $[0,T)$.

First, we derive bounds on the term, 
$1/y(\cdot,t)$. Due to the lower bound on the initial condition $y_0$, and corollary  \ref{CorII}, it follows that, there exists a positive constant $d>0$ such that 

\begin{equation}
\label{eq:bd1}
y(\cdot,t) \geq d
\end{equation}
 for all $t\in [0,T)$. This gives us the uniform upper bound $1/d$  on the term $1/y(\cdot,t)$.

Next, we consider the term,  $y_x(\cdot, t)$. We note that $y_{0} \in \mathcal{D}(A_a)$. Hence, we can apply the estimates in lemma \ref{cvglma} to get
\begin{equation}
\label{eq:bd2}
 \|y_x(\cdot,\sum_{i=1}^{m}\frac{1}{n^2})\|_{H^1} \leq \tilde{M}e^{-\alpha \lambda \sum_{i=1}^{m}\frac{1}{n}} 
\end{equation}
 for some strictly positive constants, $\tilde{M}$. Here, we have implicitly used the fact that $a$ is twice weakly differentiable and the equivalence between the norm, $\|\cdot\|_g$ and the norm, $\|\cdot\|_{H^2}$. From this it follows that $\|y_x(\cdot,t)\|_{\infty}$ is uniformly bounded on the interval, $[0,T)$, due to the continuous embedding, $H^1(0,1) \hookrightarrow L^{\infty}(0,1)$ \cite{brezis2010functional}[Theorem 8.8].
 
Lastly, we need to bound the term, $\alpha m (ay)_x(\cdot,t)$. As in the estimates for $y_x(\cdot,t)$ in the above arguments, from lemma \ref{cvglma}, we have the estimate
\begin{equation}
\label{eq:bd3}
 \|\alpha m(ay)_x(\cdot,\sum_{i=1}^{m}\frac{1}{n^2}\|_{H^1} \leq \alpha m\tilde{M}_{1}e^{- \alpha \lambda \sum_{i=1}^{m}\frac{1}{n}}.
\end{equation}
for some strictly positive constant, $\tilde{M}_1$.
The right-hand side in the estimate,\eqref{eq:bd3}, is not uniformly bounded for arbitrary $\alpha>0$ due to dependence on $m$. However, we note that $\lim_{m \rightarrow \infty} -ln ~m  + \sum_{i=1}^{m}\frac{1}{n} = \gamma$ where $\gamma>0$ is the {\it Euler-Mascheroni} constant \cite{finch2003mathematical}[section 1.5]. Therefore, by setting $\alpha  \geq 1/ \lambda$ the right-side becomes uniformly bounded for all $m \in \mathbb{Z}_+$.

From the estimates, \eqref{eq:bd1}-\eqref{eq:bd3} it follows that if $\alpha>0$ is large enough, then $v \in L^\infty(0,T;L^\infty(0,1))$. This concludes the proof.
\end{proof}

\vspace{1mm}

\begin{remark}
In the above lemma, the choice, $v(\cdot,t) =   \frac{y_x}{y}- \alpha m \frac{(a y)_x}{y}$ is definitely not unique. In fact any control law of the form, $v(\cdot,t) =   \frac{y_x}{y}- \alpha m^\beta \frac{(a y)_x}{y}$ for numerous other values of $\beta$ and $\alpha$ will also achieve the desired objective due to the fact that an exponential function of a variable grows faster than a polynomial function, as the variable tends to infinity. Additionally, we could also replace the parameter $m$ with a continuous function, $m(t)$ such that $\int_0^T m(\tau)d\tau = \infty$.
\end{remark}

\vspace{1mm}

From the above lemma the following corollary follows.

\begin{corollary}
\label{prlma2}
Let $y_0 \in \mathcal{D}(A_a)$ be such that $y_0 \geq c $ for some strictly positive constant, $c$. Suppose $f \in W^{2, \infty}(0,1)$ such that $f\geq s$, $a = 1/f$ and $\int_0^1f(x)dx = \int_0^1y_0(x)dx$, for some strictly positive constant, $s$. Let $T >0$ be the final time. Then there exists $v \in L^{\infty}(0,T;L^{\infty}(0,T))$ such that  the mild solution, $y$, of the PDE, \eqref{eq:Mainsys1}, satisfies $y(\cdot,T) = f$.
\end{corollary}

The above corollary follows from lemma \ref{prlma} using a straightforward scaling argument.

Now, we are ready to state and prove our main theorem, where we relax the assumptions made in the previous corollary on the initial condition $y_0$. A few comments are due before we prove this theorem. There are two main problems with extending corollary \ref{prlma2} with the same control as defined in eq. \eqref{eq:fbctrl}, for general positive initial conditions $y_0 \in L^2(0,1)$. Both issues can cause $v$ to tend to $\infty$ as $ t \rightarrow 0$. Firstly, with the same control law as in \eqref{eq:fbctrl}, $y_0$ needs to have a strict lower bound, as assumed in lemma IV.7. Otherwise the term in the denominator, $y$, causes blow up in $v(⋅,t) = \frac{y_x}{y} - \alpha  m \frac{(a y)_x}{y}$ near $t=0$. Secondly, for general $y_0 \in L^2(0,1)$ the numerator terms ($y_x$ and $(ay)_x$)) in $v(⋅,t) = \frac{y_x}{y} - \alpha  m \frac{(a y)_x}{y}$ also cause blow up near $t=0$ because it might be true that $y_0 \notin \mathcal{D}(A_a)$. These issues can be remedied, as shown in the following proof, by modifying the control in eq. \eqref{eq:fbctrl} appropriately.

\begin{theorem}
\label{maintheo}
Let $y_0 \in L^2(0,1)$ be such that $y_0 \geq 0 $ and $\int_0^1 f(x)dx=1$. Suppose $f \in W^{2,\infty}(0,1)$. Let $T>0$ be the final time. Then there exists $v \in L^\infty(0,T;L^\infty(0,1))$, such that the unique mild solution,$y$, of the PDE, \eqref{eq:Mainsys1}, satisfies $y(T)=f$.
\end{theorem}

\begin{proof}
Set $v(\cdot,t) = 0$, in \eqref{eq:Mainsys1}, for each $t \in [0, \epsilon/2]$ where $\epsilon \in (0,T)$ is small enough. Then the PDE, \eqref{eq:Mainsys1},  is the heat equation with Neumann boundary condition. From lemma \ref{Dave}, it follows that the solution, $y$, satisfies $y(\cdot, \epsilon)\geq c$ for some strictly positive constant $c$. Then for each $t \in (\epsilon/2,\epsilon]$ let $v(\cdot,t)= \frac{f_x}{f}$. Semigroups generated by elliptic operators are {\it analytic}. Hence, from regularizing properties of analytic semigroups, \cite{lunardi2012analytic}[Theorem 2.1.1], it follows that $y(\cdot,\epsilon) \in \mathcal{D}(B_f)$ where $B_f :\mathcal{D}(B_f) \rightarrow L^2(0,1)$ is the operator given by
\begin{equation}
B_fu = u_{xx} - (\frac{f_x}{f}u)_x
\end{equation}
for each $u \in {D}(B_f) =\lbrace w \in H^2(0,1); (w_x-\frac{f_x}{f}w)(0)=(w_x-\frac{f_x}{f}w)(1)=0 \rbrace$.
From theorem \ref{thm:FTA} we know this implies $y(\cdot,\epsilon) \in \mathcal{D}(A_a)$. Then the result follows from corollary \ref{prlma2}.
\end{proof}

\begin{remark}
(\textbf{The case when $f \in W^{1,\infty}(0,1)$})
Comparing theorem \ref{maintheo} and theorem \ref{asymp}, it is apparent that there is a gap in the result we have obtained. That is, while theorem \ref{asymp} states that any strictly positive, at least once differentiable function can be reached asymptotically, theorem \ref{maintheo} requires the target densities to be at least twice differentiable in order to be reachable in finite time. However, this assumption can be relaxed by modifying the argument in lemma \ref{prlma}. Particularly, if $f$ is only once differentiable then it is no longer true that $\mathcal{D}(A_a)$ is a subset of $H^2(0,1)$, in general. However, even if $u \notin H^2(0,1)$ we have that $\|(ay)_{xx}\|_{2} < \infty$. From this, and bounds on $\|ay\|_2$, it follows that $ \|(ay)_{x}\|_{\infty}  < \infty$ and hence, using the product rule, it follows that $\|y_x\|_{
\infty} < \infty $.

However, $B_f$ needs to be defined using its weak formulation in this case to make sense of the term 
$(\frac{f_x}{f}u)_x$. We avoid this issue for now and leave the more general case when $f \in W^{1,\infty}(0,1)$ for future work. 
\end{remark}

As pointed out in the last remark, using the approach in this paper, the requirement that $f \in W^{2,\infty}(0,1)$, can be relaxed. On the other hand, it is not clear how much regularity needs to be assumed when extending the technique in this paper to the case when the diffusion process evolves on a higher dimensional Euclidean space since embedding results depend on the dimension of the domain. However, if the constraint that $v$ is bounded is relaxed to admit square-integrable vector fields, then it is sufficient to establish the bounds on $\|\nabla y(\cdot,t)\|_2$ and $\|\nabla (ay)(\cdot,t)\|_2$, which immediately follows from estimates such as those in \ref{eq:bd2} and \ref{eq:bd3}. Such estimates on the $L^2$ norm of the control are not dimension dependent.

We would also like to point out that in proving controllability properities of the system \eqref{eq:Mainsys1}, we have taken advantage of the fact that diffusion enables infinite speed of propagation of the solution $y$ of the PDE \eqref{eq:Mainsys1} (lemma \ref{Dave}). Hence, the control laws constructed in this paper might need further modification if implemented in practice, since robots have limitations on their speed of movement. One possibility is to introduce 'virtual particles' that do propagate at infinite speeds and hence avoiding the division-by-zero in the control law \eqref{eq:fbctrl}. Another possibility is to have a more realistic model of noise in the system.

\section{CONCLUSION}
In this paper, we proved controllability properties of the Fokker-Planck equation with zero-flux boundary condition. In contrast to previous work, we established controllability with bounded control inputs. Our approach to establishing controllability using spectral properties of the elliptic operators under consideration is also novel. In our opinion, this provides a simpler approach to conclude controllability than methods in previous similar works. Future work will focus on extending the arguments in this paper to the case where the diffusion process evolves on higher-dimensional domains.

\section{APPENDIX}

\subsection[]{Proof of Theorem~\ref{thm:FTA} (See page~\pageref{thm:FTA})}
\fta*
\begin{proof}
We define the operator, $\tilde{A}_a :\mathcal{D}(\tilde{A}_a)  \rightarrow  L^2(0,1)$, defined by $\tilde{A}_au = a(x)u_{xx}$ for each $u \in {D}(\tilde{A}_a) =\lbrace w \in H^2(0,1); w_x(0)=w_x(1)=0 \rbrace$. 
Using the product rule for functions in Sobolev spaces, we can represent the operator, $\tilde{A}$, as $\tilde{A}u = (a u_x)_x - a_x u_x$ for each $u \in \mathcal{D}(\tilde{A}_a)$.                                                                                                                                                                                                                                                                                                                                                                                                                                                                                                                                                                                                                                                                                                                                                                                                                                                                                                                                                                                                                                                                                                                                                                                                                                                                                                                                                                                                                                                                                                                                                                                                                                                                                                                                                                                                                                                                                                                                                                                                                                                                                                                                                                                                                                                                                                                                                                                                                                                                                                                                                                                                                                                                                                                                                                                                                                                                                                                                                                                                                                                                                                                                                                                                                                                                                                                                                                                                                                                                                                                                                                                                                                                                                                                                                                                                                                                                                                                                                                                                                                                                                                                                                                                                                                                                                                                                                                                                                                                                                                                                                                                                                                                                                                                                                                                                                                                                                                                                                                                                                                                                                                                                                                                                                                                                                                                                                                                                                                                                                                                                                                                                                                                                                                                                                                                                                                                                                                                                                                                                                                                                                                                                                                                                                                                                                                                                                                                                                                                                                                                                                                                                                                                                                                                                                                                                                                                                                               

 The validity of the product rule of differentiation used above, that is $(pq)_x = p_x q+pq_x$ whenever $p,q\in H^{1}(0,1)$, can be seen by constructing an approximating sequence $p^n_x$ in $C^{\infty}(0,1)$ converging to $p$ in $H^1(0,1)$. Then it can be shown the integral $-\int_0^1 p^n(s)q(s) \phi_x(s)ds  = \int_0^1(p^n_x(s)q(s)+ p^n(s)q_x(s))\phi(s)ds$ for each $\phi \in C^\infty(0,1)$. Then taking the limit $n \rightarrow \infty $ gives us the validity of the product rule.

We define the bilinear form, $b:H^1(0,1) \times H^1(0,1) \rightarrow \mathbb{R}$, corresponding to this operator by,
\begin{equation}
b(u,\phi) = \big \langle a u_x ,\phi_x \big \rangle_2 +  \big  \langle a_xu_x,\phi \big \rangle_2
\end{equation}
for each $u,\phi \in H^1(0,1)$. $A_a$ is related to $b$, by 
\begin{equation}
\langle \tilde{A}_au,\phi \rangle_2= -b(u,\phi)
\end{equation}
for all $u \in \mathcal{D}(\tilde{A}_a)$ and all $\phi \in H^1(0,1)$. Using the above representation of $\tilde{A}_a$, from \cite{ouhabaz2009analysis}[Corollary 4.3] it follows that since $b$ is an {\it accretive}, closed and continuous bilinear form on $H^1(0,1)$, the associated operator, $\tilde{A}_a$ generates a {\it positivity preserving} semigroup, $(S(t))_{t \geq 0}$, on $L^2(0,1)$ such that for each $y_0 \in L^2(0,1)$ in \eqref{eq:varsys1}, the unique  mild solution of the PDE can be represented by $y(\cdot,t)= S(t)y_0$. Note that the above mentioned properties of the bilinear form have been established in \cite{ouhabaz2009analysis}. By {\it positivity preserving} we mean that if $y_0 \geq 0$ then $S(t)y_0 \geq 0$ for each $t \in [0,\infty)$.
Additionally, we note that $\tilde{A}_a$ has a eigenvalue at $0$. The function $\mathbf{1}$, defined by $\mathbf{1}(x) = 1$ for almost every $x \in (0,1)$, is an eigenvector corresponding to this eigenvalue, $0$. Hence, if $y_0 \geq c>0$ for some strictly positive parameter, $c$. Then $S(t)y_0 = S(t)(c\mathbf{1}+y_0-c\mathbf{1}) = c\mathbf{1}+S(t)(y_0-c\mathbf{1})$, for each $t \in [0, \infty)$. Since $(S(t))_{t\geq 0}$ is positivity preserving, it follows that $S(t)(y_0-c\mathbf{1}) \geq 0$ and hence $c$ is a lower bound on the solution $S(t)y_0$ if $c$ is a lower bound on the initial condition, $y_0$.
\end{proof}

\subsection[]{Proof of Theorem~\ref{thm:FTA2} (See page~\pageref{thm:FTA2})}
\ftak*
\begin{proof}
We only prove that $\mathcal{D}(B_f) = \mathcal{D}(A_a)$. The proof for the other statements follows the same line of arguments as in theorem \ref{thm:FTA}. 

The result, $\mathcal{D}(B_f) = \mathcal{D}(A_a)$, follows from the quotient rule applied at the boundary. That is, $(au)_x(0)= \big ( \frac{fu_x-uf_x}{f^2} \big )(0)=0$ implies $(u_x-u\frac{f_x}{f} )(0)=0$. The quotient rule for functions in Sobolev spaces can be seen to be true by an argument similar to the one made in theorem \ref{thm:FTA} for verification of the product rule.

\end{proof}

\bibliographystyle{plain}

\bibliography{cdcref}

\end{document}